\documentclass[11pt]{amsart}
\usepackage{amsmath,amssymb,amsfonts,amsbsy}
\usepackage{tikz-cd}
\usetikzlibrary{shapes}
\usepackage{amsbsy}
\usepackage{amscd}
\usepackage{wasysym}
\usetikzlibrary{matrix,arrows,decorations.pathmorphing}
\usepackage{graphicx}
\usepackage{float}
\usepackage{setspace}
\usepackage[margin=1.0in]{geometry}

\newtheorem{thm}{Theorem}[section]
\newtheorem{lem}[thm]{Lemma}

\newtheorem{prop}[thm]{Proposition}

\theoremstyle{definition}
\newtheorem{example}[thm]{Example}
\theoremstyle{definition}
\newtheorem{defn}[thm]{Definition}
\theoremstyle{definition}
\newtheorem{remark}[thm]{Remark}
\newtheorem{remarks}[thm]{Remarks}

\newcommand{\tr}{\mathrm{tr}}
\newcommand{\lamdiag}{\Lambda^{\mathrm{diag}}}

\newcommand{\rmv}[1]{}

\def\IR{{\bf R}}

\def\bA{{\mathbf {A}}}

\def\bbD{{\mathbf {D}}}

\def\ba{{\bf a}}

\def\b0{{\bf 0}}

\def\w{{\bf w}}

\def\x{{\bf x}}

\def\cD{{\mathcal D}}

\def\cP{{\mathcal P}}

\def\){{\right)}}
\def\({{\left(}}

\def\ket#1{| #1 \rangle}
\def\bra#1{\langle #1 |}
\def\kb#1#2{|#1\rangle\!\langle #2 |}

\newcommand{\diag}{\operatorname{diag}}
\newcommand{\spn}{\operatorname{span}}
\newcommand{\conv}{\operatorname{conv}}

\begin{document}

\title[Higher Rank Matricial Ranges and Hybrid Quantum Error Correction]{Higher Rank Matricial Ranges \\ and Hybrid Quantum Error Correction}
\author[N.~Cao, D.~W. Kribs, C.-K.~Li, M.~I. Nelson, Y.-T.~Poon, B. Zeng]{Ningping~Cao$^{1,2}$, David~W.~Kribs$^{1,2}$, Chi-Kwong~Li$^3$, Mike~I.~Nelson$^1$, Yiu-Tung~Poon$^{4,6}$, Bei~Zeng$^{1,2,5}$}

\address{$^1$Department of Mathematics \& Statistics, University of Guelph, Guelph, ON, Canada N1G 2W1}
\address{$^2$Institute for Quantum Computing, University of Waterloo, Waterloo, ON, Canada N2L 3G1}
\address{$^3$Department of Mathematics, College of William \& Mary, Williamsburg, VA, USA 23187}
\address{$^4$Department of Mathematics, Iowa State University, Ames, IA, USA 50011}
\address{$^5$Department of Physics, The Hong Kong University of Science and Technology, Clear Water Bay, Kowloon, Hong Kong}

\address{$^6$Center for Quantum Computing, Peng Cheng Laboratory, Shenzhen, 518055, China}

\begin{abstract}
We introduce and initiate the study of a family of higher rank matricial ranges, taking motivation from hybrid classical and quantum error correction coding theory and its operator algebra framework. In particular, for a noisy quantum channel, a hybrid quantum error correcting code exists if and only if a distinguished special case of the joint higher rank matricial range of the error operators of the channel is non-empty. We establish bounds on Hilbert space dimension in terms of properties of a tuple of operators that guarantee a matricial range is non-empty, and hence additionally guarantee the existence of hybrid codes for a given quantum channel. We also discuss when hybrid codes can have advantages over quantum codes and present a number of examples.
\end{abstract}



\maketitle

\section{Introduction}

For more than a decade, numerical range tools and techniques have been applied to problems in quantum error correction, starting with the study of higher-rank numerical ranges \cite{choi2005quantum,choi2006higher2} and broadening and deepening to joint higher-rank numerical ranges and beyond \cite{woerdeman2008higher,li2007higher,choi2008geometry,li2008canonical,li2009condition,kribs09resprobs,
li2011generalized,lau2018convexity}. These efforts have made contributions to coding theory in quantum error correction and have also grown into mathematical investigations of interest in their own right.  In this paper, we expand on this approach to introduce and study a higher rank matricial range motivated both by recent hybrid coding theory advances  \cite{grassl2017codes,klappen2018} and the operator algebra framework for hybrid classical and quantum error correction \cite{beny2007generalization,beny2007quantum}.
Our primary initial focus here is on a basic problem for the matricial ranges, namely, how big does a Hilbert space need to be to guarantee the existence of a non-empty matricial range of a given type, without any information on the operators and matrices outside of how many of them there are. As such, we generalize a fundamental result from quantum error correction \cite{KLV,li2011generalized} to the hybrid error correction setting.

The theory of quantum error correction (QEC) originated at the interface between quantum theory and coding theory in classical information transmission and is at the heart of designing those fault-tolerant architectures~\cite{shor1995pw,steane1996error,gottesman1996d,bennett1996ch,knill1997knill}.  It was recognized early on during these investigations that the simultaneous transmission of both quantum and classical information over a quantum channel could also be considered, most cleanly articulated in operator algebra language in \cite{kuperberg2003capacity}. More recently, but still over a decade ago, the framework of ``operator algebra quantum error correction'' (OAQEC) \cite{beny2007generalization,beny2007quantum} was introduced. Motivated by a number of considerations, including a generalization of the operator quantum error correction approach \cite{kribs2005unified,kribs2005operator} to infinite-dimensional Hilbert space \cite{beny2009infinite}, it was also recognized that the OAQEC approach could provide a framework for error correction of hybrid classical and quantum information, though this specific line of investigation remained dormant for lack of motivating applications at the time. Moving to the present time and over the past few years, advantages in addressing the tasks of transmitting both quantum and classical information together compared to independent solutions have been discovered, from both information theoretic and coding theoretic points of view \cite{devetak2005capacity,hsieh2010entanglement,yard2005simultaneous,kremsky2008classical,braunstein2011zero,patra2013algebraic,grassl2017codes,klappen2018}. Additionally it is felt that these hybrid codes may find applications in physical descriptions of joint classical-quantum systems, in view of near-future available quantum computing devices ~\cite{xiang2013hybrid} and the so-called
Noisy Intermediate-Scale Quantum (NISQ) era of computing~\cite{preskill2018quantum}.

This paper is organized as follows. In the next section we introduce the joint higher rank matricial ranges and we prove the Hilbert space dimension bound result. The subsequent section considers a special case that connects the investigation with hybrid quantum error correction; specifically, for a noisy quantum channel, our formulation of the joint higher rank matricial ranges for the channel's error or ``Kraus'' operators leads to the conclusion that a hybrid quantum error correcting code exists for the channel if and only if one of these joint matricial ranges associated with the operators is non-empty. As a consequence of the general Hilbert space dimension bound result we establish generalizations of a fundamental early result in the theory of QEC \cite{KLV,li2011generalized} to the hybrid setting. In the penultimate section we explore how hybrid error correction could provide advantages over usual quantum error correction based on this analysis. We consider a number of examples throughout the presentation and we conclude with a brief future outlook discussion.

\section{Higher Rank Matricial Ranges}

\begin{defn}\label{maindefn}
Given positive integers $m,n,p,k,K \geq 1$, let $\cP_K$ be the set of $n\times K$ rank-$K$ partial isometry matrices, so $V^*V= I_K$ for $V \in \cP_K$, and let $\cD_p$ be the set of diagonal matrices
inside the set of $p\times p$ complex matrices $M_p$. Define the {\it joint rank $(k:p)$-matricial range}
of an $m$-tuple of matrices $\bA = (A_1, \dots, A_m) \in M_n^m$  by
$$\Lambda_{(k:p)}(\bA) = \{(D_1, \dots, D_m) \in \cD_p^m:
\exists V \in \cP_{kp} \hbox{ such that }
V^*A_j V =  D_j \otimes  I_k\hbox{ for } j = 1, \dots, m \}.$$
\end{defn}

Observe that when $p=1$, $\Lambda_{(k:p)}(\bA)$ becomes the rank-$k$ (joint when $m\geq 2$) numerical range considered in \cite{choi2005quantum,choi2006higher2,woerdeman2008higher,li2007higher,choi2008geometry,li2008canonical,li2009condition,li2011generalized,lau2018convexity} and elsewhere. 

We first discuss two reductions that we can make without loss of generality. 

\begin{remark}
Since every $A\in M_n$ has a  Hermitian decomposition $A=A_1+iA_2$, with $A_1, A_2\in H_n$, the set of $n\times n$ Hermitian matrices, we only need to consider $\Lambda_{(k:p)}(\bA)$ for $\bA\in H_n^m$, where $H_n^m$ is the set of $m$-tuples of $n \times n$ Hermitian matrices.

Furthermore, suppose $T = (t_{ij}) \in M_m$ is a real
invertible matrix, and $(c_1, \dots, c_m) \in \IR^{1\times m}$.
Let $\tilde \bA = (\tilde A_1, \dots, \tilde A_m)$, where for $j = 1, \dots, m$,
$$\tilde A_j = \sum_{\ell =1}^m t_{\ell,j} A_\ell + c_j I_n . 
$$
Then one readily shows that
$(D_1, \dots, D_m) \in \Lambda_{(k:p)}(\bA)$
if and only if
$(\tilde D_1, \dots, \tilde D_m) \in \Lambda_{(k:p)}(\tilde \bA)$,
where
$\tilde D_j = \sum_{\ell =1}^m t_{\ell,j} D_\ell + c_j I_k$ for $j = 1, \dots, m$.
So, the geometry of $\Lambda_{(k:p)}(\bA)$ is completely determined by
$\Lambda_{(k:p)}(\tilde \bA)$.

Now, we can choose a suitable $T = (t_{ij})$ and $(c_1, \dots, c_m)$
so that $\{\tilde A_1, \dots, \tilde A_r, I_n\}$ is linearly independent,
and $\tilde A_{r+1} = \cdots = \tilde A_m = 0_n$. Then the geometry of
$\Lambda_{(k:p)}(\tilde \bA)$ is completely determined by
$\Lambda_{(k:p)}(\tilde A_1,\dots, \tilde A_r)$. Hence, in what follows, we always assume that
$\{A_1, \dots, A_m, I_n\}$ is linearly independent.
\end{remark}

The result we prove below is a generalization of the main result from \cite{KLV}, which applies to the $p=1$ case in our notation. This was also proved in \cite{li2011generalized} via a matrix theoretic approach and we make use of this in our proof, so we state the original result as it was presented in  \cite{li2011generalized} using the present notation. 

\begin{lem} \label{lemma1}
Let $\bA = (A_1,\ldots ,A_m) \in H_n^m$ and let $m\ge 1$ and $ k> 1$.
If
$$ n\ge (k-1) (m+1)^2,$$
then $\Lambda_{(k:1)}(\bA) \ne \emptyset$.
\end{lem}

Observe that if $(a_1,\dots,a_m)\in\Lambda_{kp}(\bA)$, then $(a_1I_p,\dots,a_mI_p)\in\Lambda_{(k:p)}(\bA)$.  Thus by Lemma~\ref{lemma1}, if $n\ge (kp-1)(m+1)^2$, then $\Lambda_{kp}(\bA)\neq\emptyset$;  hence $\Lambda_{(k:p)}(\bA)\neq\emptyset$. The following theorem gives an improvement on this bound.

\begin{thm}  \label{main1}
Let $\bA = (A_1,\ldots ,A_m) \in H_n^m$ and let $m, p \ge 1$ and $ k> 1$.
If
$$ n\ge (m+1)( (m+1)(k-1)+k(p-1)),$$
then $\Lambda_{(k:p)}(\bA) \ne \emptyset$.
\end{thm}

\begin{proof} We will prove the result by induction on $p$.
When $p=1$, the bound $(m+1)( (m+1)(k-1)+k(p-1))=(k-1)(m+1)^2$ is given in Lemma~\ref{lemma1}. 

Suppose $p>1$. 
We suppose for $r<p$, we can find an $n\times rk$ 
matrix $U_r$    and $r\times r$ diagonal matrices $D_{j,r} $, $1\le j\le m$  such that 
$U_r^* U_r  = I_{rk}$ and $U_r^*A_jU_r  =  D_{j,r}   \otimes I_k$ for all  $1\le j\le m$.

Since $(m+1)( (m+1)(k-1)+k(p-1))>(k-1)(m+1)^2$, there exist an $n\times k$ matrix $U_1$  and scalars $d_j$, $1\le j\le m$ such that $U_1^* U_1  = I_k$ and $U_1^*A_jU_1  =d_jI_k$ for all  
$1\le j\le m$.
Let $U$ be unitary with first $\ell$ columns containing the columns spaces of
$U_1, AU_1, \dots, A_m U_1$. Then $\ell \le (m+1)k$ and 
$U^*A_jU = B_j \oplus C_j$ with $B_j \in M_\ell$ for $j = 1, \dots, m$, and 
$C_j \in M_{n-\ell}$, where $n-\ell \ge (m+1)((m+1)(k-1) + k(p-2))$.
By induction assumption, 
$\Lambda_{(k:p-1)}(C_1,\dots, C_m)$ is non-empty, say, containing an
$m$-tuple of diagonal matrices
$(D_{j1}, \dots, D_{jm}) \in M_{p-1}^m$. So,
we can find an $n\times (k-1)(m+1)^2$ matrix $U_2$ such that 
$U_2^*A_jU_2 = D_{j\ell} \otimes I_k$ for $j = 1, \dots, m$.
Thus, there is  $V = [U_1|V_2] \in M_{n,pk}$
such that $V^*V= I_{pk}$ and 
$V^*A_jV = d_j I_k \oplus D_{j\ell} \otimes I_k$ for $j = 1, \dots, m$.
Hence, $\Lambda_{(k: p)}(\bA)\neq \emptyset$.
\end{proof}

\begin{remarks} 

$(i)$ Let $n(k,m)$ (respectively, $n(k:p,m)$) be the minimum number such that for all $n\ge n(k,m)$ (respectively, $n(k:p,m)$), we have $\Lambda_{k}(\bA) \ne \emptyset$  (respectively, we have $\Lambda_{(k:p)}(\bA) \ne \emptyset$) for all $\bA\in H_n^m$. 
 Clearly, we have $n(kp,m)\ge n(k:p,m)\ge kp$. In  Example \ref{x2} and  \ref{xyz2}, we will see that sometimes the lower bound can be attained.

$(ii)$ The upper  bound $(m+1)( (m+1)(k-1)+k(p-1))\ge n(k:p,m)$ in Theorem \ref{main1}   is not optimal.  The  same proof   shows that $n(k:p,m)\le n(k,m) +(m+1)k(p-1)$. So, if we can lower the bound for $n(k,m)$, then we can lower the bound for $n(k:p,m)$. For example, since $n(k,1)=2k-1$ \cite{choi2006higher2} and $n(k,2)=3k-2$ \cite{li2011generalized}, we have $n(k:p,1)\le 2pk-1$ and $n(k:p,2)\le 3pk-2$. We also note that using Fan and Pall's interlacing theorem \cite{FP}, one can show that  $n(k:p,1)=(p+1)k-1$.

$(iii)$ In the proof of Theorem~\ref{main1}, suppose $U_1^*A_j U_1$ has leading $k\times k$ submatrix
equal to $a_{j1} I_k$. Then we can find a unitary $X$ such that
$X^*A_jX = (B^{(j)}_{pq})_{1\le p, q \le 2}$  with $B^{(j)}_{11} = a_j I_k$,
$B^{(j)}_{12} = 0_{k\times r}$ and $B^{(j)}_{13}$ is $k \times s$ with $s \le mk$.
That is why we can induct on the leading $(n-s)\times (n-s)$ matrices.
Of course, we can have some savings if $s < mk$ at any step.

$(iv)$  Also, when $m = 1$, it does not matter whether we want $ D_j   \otimes   I_k $ or 
$C_j\otimes I_k$ for diagonal $D_j$ or general Hermitian $C_j$. We can diagonalize
$C_j$. Note that if $n = (p+1)k-1$, then the set $\Lambda_{(k:p)}(A)$ is unique
if the eigenvalues of $A$ are distinct. It should be possible to say more 
if there are repeated eigenvalues, and in that case one can lower the requirement of $n \ge (p+1)k-1$.

$(v)$ When $\{A_1, \dots, A_m\}$ is a commuting family, then $A_p + iA_q$
is normal for any $p< q$. The results in \cite{GLPS} might be useful to study this further. 

$(vi)$ One could also study a more general class of matricial ranges in which Definition~\ref{maindefn} would be viewed as a special case; namely, the definition could be broadened to allow for arbitrary $p \times p$ matrices in the $m$-tuples of $\Lambda_{(k:p)}(\bA)$, removing the diagonal matrix restriction. One can generalize Theorem~\ref{main1} and  obtain other interesting results; see Section 5.
\end{remarks}

Using similar techniques, we may also apply recent related work on the shape of joint matricial ranges \cite{lau2018convexity} to obtain the following. 

\begin{thm} \label{main2}
Let $\bA = (A_1, \dots, A_m)\in H_n^m$, and $m,p, k \geq 1$ satisfy
$n \ge (kp(m+2)-1)(m+1)^2$. Then 
$\Lambda_{(kp(m+2):1)}(\bA)$ is non-empty, and
 $\Lambda_{(k:p)}(\bA)$ is star-shaped with star center
$(a_1I_p, \dots, a_m I_p)$ for any $(a_1, \dots, a_m) \in \Lambda_{(kp(m+2):1)}(\bA)$,
i.e., $t (a_1 I_{p}, \dots, a_m I_{p}) + (1-t)(B_1,\dots, B_m) \in \Lambda_{(k:p)}(\bA)$
for any $t \in [0,1]$ and $(B_1, \dots, B_m) \in \Lambda_{(k:p)}(\bA)$.
\end{thm}

\section{Application to Hybrid Quantum Error Correction}

In quantum information,  a {\it quantum channel}
 corresponds to a completely positive and trace preserving (CPTP) linear
map $\Phi: M_n \rightarrow M_n$. By the structure theory of
such maps \cite{choi1975completely},
there is a finite set $E_1, \dots \in M_n$
with $\sum_{j}   E_j^*E_j = I_n$ such that for all $\rho \in M_n$,
\begin{equation}\label{channel}
\Phi(\rho) = \sum_{j} E_j \rho E_j^*.
\end{equation}
These operators are typically referred to as the {\it Kraus operators} for $\Phi$ \cite{nielsen2000quantum}, and the minimal number of operators $E_j$ required for this operator-sum form of $\Phi$ is called the {\it Choi rank} of $\Phi$, as it is equal to the rank of the Choi matrix for $\Phi$ \cite{choi1975completely}. In the context of quantum error correction, $E_j$ are viewed as the {\it error operators} for the physical noise model described by $\Phi$.

The OAQEC framework \cite{beny2007generalization,beny2007quantum} relies on the well-known structure theory for finite-dimensional von Neumann algebras (equivalently C$^*$-algebras) when applied to the finite-dimensional setting \cite{davidson1996c}. Specifically, codes are characterized by such algebras, which up to unitary equivalence can be uniquely written as $\mathcal A = \oplus_i (I_{m_i}\otimes M_{n_i})$. Any $M_{n_i}$ with $n_i > 1$ can encode quantum information; which when $m_i=1$ corresponds to a standard (subspace) error-correcting code \cite{shor1995pw,steane1996error,gottesman1996d,bennett1996ch,knill1997knill} and when $m_i > 1$ corresponds to an operator quantum error-correcting subsystem code \cite{kribs2005unified,kribs2005operator}. If there is more than one summand in the matrix algebra decomposition for $\mathcal A$, then the algebra is a hybrid of classical and quantum codes. It has been known for some time that algebras can be used to encode hybrid information in this way \cite{kuperberg2003capacity}, and OAQEC provides a framework to study hybrid error correction in depth. Of particular interest here, we draw attention to the recent advance in coding theory for hybrid error-correcting codes \cite{grassl2017codes}, in which explicit constructions have been derived for a distinguished special case of OAQEC discussed in more detail below.

In the Schr\"{o}dinger picture for quantum dynamics, an OAQEC code is explicitly described as follows: $\mathcal A$ is {\it correctable} for $\Phi$ if there is a CPTP map $\mathcal R$ such that for all density operators $\rho_i \in M_{n_i}$, $\sigma_i \in M_{m_i}$ and probability distributions $p_i$,  there are density operators $\sigma_i^\prime$ such that
\[
(\mathcal R \circ \Phi)\, \left( \sum_i p_i  (\sigma_i \otimes \rho_i )    \right)  =    \sum_i p_i  (\sigma_i^\prime \otimes \rho_i) .
\]
This condition is perhaps more cleanly phrased in the corresponding Heisenberg picture as follows: $\mathcal A$ is correctable for $\Phi$ if there is a channel  $\mathcal R$ such that for all $X\in \mathcal A$,
\[
(\mathcal P_{\mathcal A} \circ \Phi^\dagger \circ \mathcal R^\dagger ) \, (X)\, = \, X ,
\]
where $\Phi^\dagger$ is the Hilbert-Schmidt dual map (i.e.,  $\tr( X \Phi(\rho)) = \tr (\Phi^\dagger(X) \rho)$) and $\mathcal P_{\mathcal A}(\cdot) = P_{\mathcal A} (\cdot) P_{\mathcal A}$ with $P_{\mathcal A}$ the unit projection of $\mathcal A$.

From \cite{beny2007generalization,beny2007quantum},  we have the following useful operational characterization of correctable algebras in terms of the Kraus operators for the channel:

\begin{lem}\label{testablelemma}
An algebra $\mathcal A$ is correctable for a channel $\Phi(\rho) = \sum_i E_i \rho E_i^*$ if and only if
\begin{equation}\label{testableconds}
[P E_i^* E_j P , X ] =0 \quad \quad \forall X \in \mathcal A,
\end{equation}
where $P$ is the unit projection of $\mathcal A$.
\end{lem}

In other words, $\mathcal A$ is correctable for $\Phi$ if and only if the operators $PE_i^* E_j P$ belong to the commutant $P \mathcal A^\prime P = P \mathcal A^\prime = \mathcal A^\prime P$. Applied to the familiar case of standard quantum error correction, with $\mathcal A = M_k$ for some $k$, we recover the famous Knill-Laflamme conditions \cite{knill1997knill}: $\{ P E_i^* E_j P \}_{i,j} \subseteq \mathbb{C} P$. The result applied to the case $\mathcal A = I_m \otimes M_k$ yields the testable conditions from operator quantum error correction \cite{kribs2005unified,kribs2005operator}:  $\{ P E_i^* E_j P \}_{i,j} \subseteq M_m \otimes I_k$. Anything else involves direct sums and has a hybrid classical-quantum interpretation as noted above.

We next turn to the distinguished special hybrid case noted above. First some additional notation: we shall assume all our channels act on a Hilbert space $\mathcal H$ of dimension $n \geq 1$, and so we may identify $\mathcal H = \mathbb{C}^n$ and let $\{ \ket{e_i} \}$ be the canonical orthonormal basis. Our algebras $\mathcal A$ then are subalgebras of the set of all linear operators $\mathcal L(\mathcal H)$ on $\mathcal H$, which in turn is identified with $M_n$ through matrix representations in the basis $\ket{e_i}$. We shall go back and forth between these operator and matrix perspectives as convenient.

As in \cite{grassl2017codes}, consider the case that $\mathcal A = \oplus_r \mathcal A_r$ with each $\mathcal A_r = M_k$ for some fixed $k\geq 1$. Let us apply the conditions of Eq.~(\ref{testableconds}) to  such algebras. Let $P_r$ be the (rank $k$) projection of $\mathcal H$ onto the support of $\mathcal A_r$, so that the $P_r$ project onto mutually orthogonal subspaces and $P = \sum_r P_r$ is the unit projection of $\mathcal A = \oplus_r P_r \mathcal L(\mathcal H) P_r$. Observe here the commutant of $\mathcal A$ satisfies: $P\mathcal A^\prime P = \oplus_r \mathbb{C} P_r$. Thus by Lemma~\ref{testablelemma}, it follows that $\mathcal A$ is correctable for $\Phi$ if and only if for all $i,j$ there are scalars $\lambda_{ij}^{(r)}$ such that
\begin{equation}\label{mtxalgeqns}
P E_i^* E_j P = \sum_r \lambda_{ij}^{(r)} P_r,
\end{equation}
which is equivalent to the equations:
\begin{equation}\label{testablematrixalg}
P_r  E_i^* E_j P_s = \delta_{rs} \lambda_{ij}^{(r)} P_r  \quad \quad \forall r,s.
\end{equation}
Indeed, these are precisely the conditions derived in \cite{grassl2017codes} (see Theorem~4 of \cite{grassl2017codes}).

For what follows, let $\mathcal V_r$, $1\le r \le p$ be mutually orthogonal $k$-dimensional subspaces of $\mathbb{C}^n$ and $P_r$
the orthogonal projection of
$\mathbb{C}^n$ onto $\mathcal V_r$, for $1\le r \le p$. Following \cite{grassl2017codes}, we say that $\{\mathcal V_r :1\le i\le p \}$
is a {\it hybrid $(k:p)$ quantum error correcting code} for the quantum channel
$\Phi$ if for all $i,j$ and all $r$ there exist scalars $\lambda_{ij}^{(r)}$ such that Eqs.~(\ref{mtxalgeqns}) are satisfied.

Consideration of the matricial ranges defined above is motivated by the following fact, which can be readily verified from Eqs.~(\ref{testablematrixalg}).

\begin{lem}
A quantum channel $\Phi$ as defined in Eq.~(\ref{channel}) has a hybrid error correcting code of dimensions $(k:p)$
if and only if
$$\Lambda_{(k:p)}(E_1^*E_1, E_1^*E_2, \dots, E_r^*E_r) \ne \emptyset.$$
\end{lem}

We note that given a rank-$kp$ projection, with say $P = \sum_{i=1}^{kp} \kb{e_i}{e_i}$, and diagonal matrices $D_j$ that make $\Lambda_{(k:p)}$ nonempty, we may define the desired projections for $1 \leq r \leq p$, by $P_r = \sum_{i=1}^{k} \kb{e_{(r-1)k+ i}}{e_{(r-1)k+ i}}$.

\begin{thm}  \label{mainhybridqec}
Let $\Phi$ be a quantum channel as defined in Eq.~(\ref{channel}) with Choi rank equal to $c$. Then $\Phi$ has a hybrid error correcting code of dimensions $(k:p)$ if
 $$ \dim \mathcal H \ge c^2(c^2(k-1)+k(p-1)).$$
\end{thm}

\begin{proof}
Suppose  $\{ E_1, \ldots , E_c\}$ is a minimal set of Kraus operators  that implement $\Phi$ as in (\ref{channel}).  For  $1\le j< \ell\le c$, let 
$F_{j\ell}=\frac{1}{2}(E_j^* E_\ell+E_\ell^* E_j)$ and $F_{\ell j}=\frac{1}{2i}(E_j^* E_\ell-E_\ell^* E_j)$. Also, let $F_{jj}= E_j^* E_j$ for $1\le j \le c$. 
Since $\sum_{j=1}^cE_j^*E_j=I$, the operator subspace $\spn \{ F_{j\ell}:1\le j,\  \ell\le c\}$ has a basis $\{A_0=I,\dots,A_m\}$ with $m\le  c^2-1$.  The result now follows from an application of Theorem~\ref{main1}.
\end{proof}

Theorem \ref{mainhybridqec} is useful if we have no information about the $E_i'$s, except the number $c$. If the $E_i'$s are given, we may get  a hybrid code even when $n$ is lower than the bound given in Theorem \ref{main1} or  \ref{mainhybridqec}. The saving can come from two sources: 1) The   subspace spanned by $\{E_i^*E_j:1\le i,j\le c\}$ can have dimension (over $\mathbb{R}$) smaller than $c^2$ in particular when restricted to the code, or 2) the  operators $\{E_i^*E_j\}$ have some specific structures. 
We give some examples to demonstrate this.

\begin{example}\label{x2}
 Consider the error model on a three-qubit system
    \[   \Phi(\rho) =   p (X_2 \rho X_2) + (1 - p) \rho, \]
where $X_2 = I \otimes X \otimes I$ and $X$ is the Pauli bit flip operator and $0 < p < 1$ is some fixed probability. It is not hard to see that the codes $C_1 = \mbox{span}\, \{ |000 \rangle , |001\rangle    \} $ and $C_2 = \mbox{span}\, \{ |100 \rangle , |101\rangle    \} $ together form a correctable hybrid code for $\Phi$. One would seek to examine the matricial range $$\Lambda_{(k:p)} (E_1^*E_1, E_1^* E_2, E_2^* E_1, E_2^*E_2) = \Lambda_{(k:p)}(I,X_2,X_2,I).$$
By the above reduction to linearly independent sets of Kraus operators, we would be interested in the geometry of $\Lambda_{(k:p)}(X_2)$. Since $X_2$ is unitarily similar to $I_4\oplus -I_4$,  $\Lambda_{(4:2)}(X_2) =\{\diag(1,-1)\}$. Thus, for  this example,  we have $m=1, k=4, p=2, n = 8 and c=2$.  
\end{example}

\begin{example}\label{xyz1}

  Consider the quantum channel on a three-qubit
system given by
    \[ \Phi(\rho) = p_0 \rho + p_1 X^{\otimes 3}\rho {X^{\otimes 3}}^*+  p_2 Y^{\otimes 3}\rho {Y^{\otimes 3}}^* +  p_3 Z^{\otimes 3}\rho {Z^{\otimes 3}}^*, \]
where $p_0, \dots , p_3$ are probabilities summing to $1$ and
$X^{\otimes 3} = X \otimes X \otimes X$ etc, with the Pauli
matrices $X,Y,Z$.

In this case the relevant operators $E_i^* E_j$ form the 3-tuple
$(X^{\otimes 3}, Y^{\otimes 3}, Z^{\otimes 3})$, and we set $m=3$,
$k=4, p=1$. Defining a partial isometry $V : \mathbb{C}^4 \rightarrow
\mathbb{C}^2 \otimes \mathbb{C}^2 \otimes \mathbb{C}^2$ by
\[
V = \kb{000}{00} + \kb{011}{01} + \kb{101}{10} + \kb{110}{11},
\]
one can verify that $V^* V = I_4$ and
\begin{eqnarray*}
V^* (X^{\otimes 3}, Y^{\otimes 3}, Z^{\otimes 3}) V &=& (0_2 \otimes
I_2, 0_2 \otimes I_2,  I_2 \otimes I_2) \\ &=& (0_4, 0_4,I_4).
\end{eqnarray*}

Therefore, $\Lambda_4 (X^{\otimes 3}, Y^{\otimes 3}, Z^{\otimes 3}) \neq \emptyset$. However,  $\Lambda_{(4:2)} (X^{\otimes 3}, Y^{\otimes 3}, Z^{\otimes 3})= \emptyset$ because 
  $X^{\otimes 3}$ and $ Y^{\otimes 3}$ do not commute. 
\end{example}

\begin{example}\label{xyz2} Extend the previous example  to  a four-qubit
system given by
    \[ \Psi(\rho) = p_0 \rho + p_1 X^{\otimes 4}\rho {X^{\otimes 4}}^*+  p_2 Y^{\otimes 4}\rho {Y^{\otimes 4}}^* +  p_3 Z^{\otimes 4}\rho {Z^{\otimes 4}}^*, \]
where $p_0, \dots , p_3$ are probabilities summing to $1$. 

In this case the relevant operators $E_i^* E_j$ form the 3-tuple
$(X^{\otimes 4}, Y^{\otimes 4}, Z^{\otimes 4})$, and we set $m=3$. We are going to show that there is a unitary matrix $U\in M_{16}$ such that 
\begin{equation}\label{xyz3}
U^*X^{\otimes 4}U = D_X\otimes I_4,\quad U^*Y^{\otimes 4}U = D_Y\otimes I_4,\quad U^*Z^{\otimes 4}U = D_Z\otimes I_4,  
\end{equation}
for some diagonal matrices $D_X,D_Y,D_Z\in M_4$. Hence, we will have $\Lambda_{(4:4)} (X^{\otimes 4}, Y^{\otimes 4}, Z^{\otimes 4}) \neq \emptyset$. In this case, 
$k=4, p=4$ and $n=16=kp$. Thus, the smallest possible $n$ is also achieved.

For $J=(j_1j_2j_3j_4)\in \{0,1\}^4$, let $\ket{J}=\ket{j_1j_2j_3j_4}$ and $|J|=\sum_{i=1}^4j_i$.  Since $Y_4\ket{J}=(-1)^{|J|}X_4\ket{J}$, we have

\begin{equation}\label{xyz4}
\begin{array}{rl}
X_4(\ket{J} +X_4\ket{J})&=\ket{J} +X_4\ket{J}\\&\\
X_4(\ket{J} -X_4\ket{J})&=-(\ket{J} -X_4\ket{J})\\&\\
Y_4(\ket{J} +X_4\ket{J})&=\left\{\begin{array}{ll} \ket{J} +X_4\ket{J}&\mbox{ if } |J|   \mbox{ is  even}  \\&\\
-( \ket{J} +X_4\ket{J})&\mbox{ if } |J|   \mbox{ is  odd}  \end{array}\right.\\&\\
Y_4(\ket{J} -X_4\ket{J})&=\left\{\begin{array}{ll}-(\ket{J} -X_4\ket{J})&\mbox{ if } |J|   \mbox{ is  even}  \\&\\
( \ket{J} -X_4\ket{J})&\mbox{ if } |J|   \mbox{ is  odd}  \end{array}\right.
\end{array}
\end{equation}

Define a unitary matrix $U=\frac{1}{2}[u_1\cdots u_{16}]$ with columns given by
$$\begin{array}{rl}
u_{1 }=&(\ket{0000  }+\ket{1111  })+(\ket{0011  }+\ket{ 1100 })\\&\\
u_{2 }=&(\ket{0000  }+\ket{1111  })-(\ket{0011  }+\ket{ 1100 })\\&\\
u_{ 3}=&(\ket{0101  }+\ket{ 1010 })+(\ket{0110  }+\ket{1001  })\\&\\
u_{4 }=&(\ket{0101  }+\ket{ 1010 })-(\ket{0110  }+\ket{1001  })\\&\\
u_{5 }=&(\ket{0001  }+\ket{1110  })+(\ket{0010  }+\ket{ 1101 })\\&\\
u_{ 6}=&(\ket{0001  }+\ket{1110  })-(\ket{0010  }+\ket{ 1101 })\\&\\
u_{7 }=&(\ket{0100  }+\ket{1011  })+(\ket{0111  }+\ket{ 1000 })\\&\\
u_{8 }=&(\ket{0100  }+\ket{1011  })-(\ket{0111  }+\ket{ 1000 })\\&\\
u_{9}=&(\ket{0000  }-\ket{1111  })+(\ket{0011  }-\ket{ 1100 })\\&\\
u_{10 }=&(\ket{0000  }-\ket{1111  })-(\ket{0011  }-\ket{ 1100 })\\&\\
u_{11}=&(\ket{0101  }-\ket{ 1010 })+(\ket{0110  }-\ket{1001  })\\&\\
u_{12}=&(\ket{0101  }-\ket{ 1010 })-(\ket{0110  }-\ket{1001  })\\&\\
u_{13 }=&(\ket{0001  }-\ket{1110  })+(\ket{0010  }-\ket{ 1101 })\\&\\
u_{ 14}=&(\ket{0001  }-\ket{1110  })-(\ket{0010  }-\ket{ 1101 })\\&\\
u_{15 }=&(\ket{0100  }-\ket{1011  })+(\ket{0111  }-\ket{ 1000 })\\&\\
u_{16 }=&(\ket{0100  }-\ket{1011  })-(\ket{0111  }-\ket{ 1000 }) 
\end{array}\,.
$$
Since, $Z_4=X_4Y_4$, by (\ref{xyz4}), we have (\ref{xyz3}) with 
\begin{equation}\label{xyz5}D_X=\diag(1,1,-1,-1),\ D_Y=\diag(1,-1,-1,1)\ \mbox{ and }\ D_Z=\diag(1,-1,1,-1)\,.\end{equation}

\end{example}

\begin{remark} 
More generally, one can consider the class of correlation channels
studied in \cite{li2011efficient}, which has error operators
$X^{\otimes n}$, $Y^{\otimes n}$, $Z^{\otimes n}$ normalized with
probability coefficients. It is proved there that when $n$ is odd,  $\Lambda_{2^{n-1}}\left(X^{\otimes n} ,  Y^{\otimes n} ,  Z^{\otimes n}\right)\neq \emptyset$. Thus 
$n$ qubit
codewords encode $(n-1)$ data qubits when $n$ is odd.  When $n$ is even, it follows that   $\Lambda_{2^{n-2}}\left(X^{\otimes n} ,  Y^{\otimes n} ,  Z^{\otimes n}\right)\neq \emptyset$.
Using a proof similar to the above example, we can show that $\Lambda_{(2^{n-2}:4)}\left(X^{\otimes n} ,  Y^{\otimes n} ,  Z^{\otimes n}\right)=\{(D_X,D_Y,D_Z)\}$, with  $D_X,D_Y,D_Z$ given by (\ref{xyz5}).  
It has been proven in \cite{li2011efficient} that  for $n$ even, 
$\Lambda_{2^{n-1}}\left(X^{\otimes n} ,  Y^{\otimes n} ,  Z^{\otimes n}\right)= \emptyset$. 
Actually, we can show that 
$\Lambda_{k}\left(X^{\otimes n} ,  Y^{\otimes n} ,  Z^{\otimes n}\right)= \emptyset$ for all 
$k>2^{n-2}$. Therefore, we can encode at most $n-2$ qubits.  Using the hybrid code, we can get $2$ 
additional classical bits. 
Very recently, this scheme has been implemented using 
IBM's quantum computing framework qiskit \cite{lilylespoon}. 
\end{remark}

\section{Exploring Advantages of Hybrid Quantum Error Correction} \label{advantage}

A straightforward way to form hybrid codes is to use quantum codes to directly transmit classical information. 
However, it is impractical since quantum resources are more expensive than classical resources.
Thus, realistically, hybrid codes are more interesting when the simultaneous transmission of classical information and quantum information do possess advantages.
One of such situations is, with a fixed set of operators $\mathbf{A}$, hybrid quantum error correcting codes exist but the corresponding quantum codes do not exist for the same system dimension $n$, i.e. $\Lambda_{(k:p) }(A) \ne \emptyset$ and $\Lambda_{kp}(A) = \emptyset$.


\begin{prop} \label{1.1}
 Suppose $A$ is an $n\times n$ Hermitian matrix with eigenvalues $a_1\ge a_2\ge\cdots\ge a_n$. Then
$$\begin{array}{rl}\Lambda_{k p }(A) =&\left\{t:a_{n+1-kp}\le  t\le a_{kp}\right\}\\&\\
\Lambda_{(k:p) }(A) =&\left\{(t_1,\dots,t_p):  a_{ik}\le t_{[i]}\le a_{n+1-(p-i+1)k}  \mbox{ for }  1\le i\le  p \right\},
\end{array}
$$
where here, $t_{[ 1]}\ge  t_{[2 ]}\ge  \cdots\ge  t_{[ n]}  $ is  a rearrangement of $t_{1 }\, t_{2 }, \cdots,  t_{ n}   $ in decreasing order.
\end{prop}

\begin{proof} 
The first statement follows from \cite{choi2006higher2}. For the second, by a result of Fan and Pall \cite{FP}, $b_1\ge b_2\cdots \ge b_m$ are the eigenvalues of $U^*AU$ for some $n\times m$ matrix $U$ satisfying $U^*U=I_m$ if and only if
$$a_i\ge b_i\ge a_{n-m+i}\quad\forall \, 1\le i\le m,$$
from which the result follows. 
\end{proof}

\begin{remark}
$(i)$ If we require the components $(t_1,\dots,t_p)$ in $\Lambda_{(k:p) }(A)$ to be in decreasing order,  then the ``ordered'' $\Lambda_{(k:p) }(A)$ is convex.

$(ii)$ $\Lambda_{kp}(A)=[a_{n+1-kp}, a_{kp}]$ is obtained by taking the convex hull of the eigenvalues of $A$ after
 deleting the $(n-kp+1)$ largest and
smallest eigenvalues. The following proposition is a generalization of this result.
\end{remark}

\begin{prop}\label{1.2}  Suppose $A_i=\diag(a^i_1,a^i_2,\dots,a^i_n)$ for  $i=1,\dots,m$ with $a^i_j\in \mathbb{R}$. Let $\ba_j=(a^1_j,a^2_j,\dots,a^m_j)$ for $j=1,\dots,n$.
For  $S\subseteq \{1,\dots,n\}$, let $X_S=\conv \{\ba_j:j\in S\}$. Then for every $1\le k\le n$,
\begin{equation}\label{eq1}
\Lambda_{k }(\bA)\subseteq \cap\{ X_S:S\subset \{1,2,\dots,n\},\ |S|=n-k+1\}. 
\end{equation}
\end{prop}

\begin{proof} It suffices to prove that $\Lambda_{k }(\bA)\subset X_S$ for $S=\{1,2,\dots,n-k+1\}$.
Suppose we have $\x=(x_1,x_2\dots,x_m)\in \Lambda_{k }(\bA)$. Then there exists a rank $k$ projection $P$ such that $PA_iP=x_iP$  for  $i=1,\dots,m$. Consider the subspace $W=\spn\{e_1,\dots,e_{n-k+1}\}$. Then there exists a unit vector $\w=(w_1,\dots,w_n)^t\in \mathbb{R}^n$ such that $Pw=\w$. Therefore, for $1\le i\le m$,
$$x_i=x_i\w^*\w=x_i\w^*P\w=\w^*PA_iP\w=w^*A_i\w=\sum_{j=1}^{n-k+1}|w_j|^2a^i_j\,' . $$
Hence, $\x=\sum_{j=1}^{n-k+1}|w_j|^2\ba_j\in X_S$.
\end{proof}

By  the result in \cite{li2008canonical}, equality holds in (\ref{eq1}) for $m =1,2$. For $m>2$, $\Lambda_{k }(\bA)$ may not be convex and equality may not hold.

\begin{prop}\label{1.3}   Let $A_i$, $ 1\le i\le m$ be as given in Proposition \ref{1.2}. Then we have: 
\begin{itemize}
\item[(1)] If $n\ge (m+1)k-m$, then $\Lambda_{k}(\bA) \neq \emptyset$. The bound $(m+1)k-m$ is best possible; i.e., if $n<(m+1)k-m$, there exist real diagonal matrices $A_1,\dots,A_m$ such that $\Lambda_{k}(\bA)=\emptyset$.

\item[(2)]  If $n\ge p( (m+1)k -m)$, then $\Lambda_{(k:p)}(\bA) \neq \emptyset$.
\end{itemize}
\end{prop}

\begin{proof} The statement (1)  follows from Tverberg's Theorem \cite{Tv} and Proposition \ref{1.2}. (Also, see Example \ref{eg1}.)

For (2), note that if $n\ge p((m+1)k -m)$, we can decompose each $A_i=\oplus_{j=1}^pA_i^j$ with $A_i^j\in M_{n_j}$, and $n_j\ge (m+1)k -m$. Then, by the result in 1), $\Lambda_{k }(A_1^j,\dots,A_m^j)\neq \emptyset$ and  the result follows. 
\end{proof}

\begin{remark}
By the above proposition, for  $p((m+1)k -m)\le n< (m+1)kp-m$, we can construct $A_1 ,\dots,A_m$ such that $\Lambda_{kp}(\bA) = \emptyset$ and   $\Lambda_{(k:p)}(\bA) \neq \emptyset$.
\end{remark}

\begin{example}\label{eg1} Suppose $p((m+1)k -m)\le n< (m+1)kp-m$. We are going to show that there exist
$A_1,\dots,A_m$ such that $\Lambda_{kp}(\bA)=\emptyset$ and $\Lambda_{(k:p)}(\bA)\neq\emptyset$.

Let $r=\left [\dfrac{n}{kp-1}\right ]$, the greatest integer $\le \dfrac{n}{kp-1}$. Then $1\le r\le (m+1)$ and $r=m+1$ if and only if $n=(m+1)(kp-1)$.  Define for $1\le i\le \min\{r,m\}$, $A_i=\diag(a^i_1,a^i_2,\dots,a^i_n)$, where $a_i^j=1$ for $(i-1)(kp-1)+1\le j \le i(kp-1)$ and $a_i^j=0$ otherwise. Then, by Proposition \ref{1.2},  $\Lambda_{kp}(\bA)=\emptyset$. Since $n\ge p((m+1)k -m)$, by Proposition \ref{1.3}, $\Lambda_{(k:p)}(\bA)\neq\emptyset$. \qed
\end{example}

\section{Outlook}

As mentioned in Remark 2.5 (vi), 
one can further extend the definition of $\Lambda_{(k:p)}(\bA)$ and consider 
$(B_1, \dots, B_m) \in M_p^m$ such that 
$V^*A_jV = B_j \otimes I_k$ for some $n\times pk$ matrix $V$ satisfying $V^*V = I_{pk}$
without requiring $B_1, \dots, B_m$ to be diagonal matrices as in Definition 2.1.
We can then use the recent results and techniques in \cite{lau2018convexity} to show that
this set is non-empty and star-shaped if $n$ is sufficiently large.
This generalization also has a potential implication to the study of quantum error correcting codes.
In particular, one may use random qubits to do the encoding and protect the data bits 
in the quantum error correction process. We plan further research in this direction.




It has been proved that transmitting classical and quantum information simultaneously provides advantages from an information-theoretic perspective~\cite{devetak2005capacity}. Practical hybrid classical-quantum error correcting codes built on the mathematical techniques introduced here that achieve these advantages could benefit various quantum communication tasks. Communication protocols based on such hybrid codes are expected to enhance the communication security or increase channel capacities. We leave these lines of investigation for future studies.

Recently, the theory of QEC, and especially the framework of OAQEC, has been found to be closely related to the AdS/CFT correspondence and holographic principle in various ways~\cite{almheiri2015bulk,pastawski2015holographic,harlow2017ryu,sanches2016holographic,fiedler2017jones,pastawski2017code}. For instance, Almheiri, Dong and Harlow interpret the complex dictionary in AdS/CFT as the encoding operations of certain operator algebra quantum error correcting codes and bulk local operators are logic operators for these error correcting codes~\cite{almheiri2015bulk}. At the same time, holographic codes also inspire new code design methods from a geometric perspective~\cite{pastawski2017code}. The matricial range approach to hybrid codes introduced here could conceivably generate new connections between QEC and the theory of quantum gravity.


\vspace{0.1in}

{\noindent}{\it Acknowledgements.} D.W.K. was supported by NSERC. C.K.L.
is an affiliate member of the Institute for Quantum Computing,
University of Waterloo. His research was supported by USA NSF grant DMS 1331021,
and Simons Foundation Grant 351047.
M.N. was supported by Mitacs and the African Institute for Mathematical Sciences.  B.Z. was supported by NSERC.

\bibliography{HybridRefs}
\bibliographystyle{unsrt}

\end{document}